%% file: main.tex
\documentclass[dvips,12pt]{article}
\usepackage{amsmath,amssymb,amsthm}
\usepackage{graphicx,subfigure,gastex}
\usepackage{cite}
\usepackage{enumerate}
\usepackage[usenames]{color}

\usepackage{varioref}
\usepackage[breaklinks=true,%
            colorlinks=true,% if false, you get a coloured box around any link and
            linkcolor=black,% if true the text itself is coloured according to linkcolor,
            anchorcolor=black,% anchorcolor,...,urlcolor.
            citecolor=black,%
            filecolor=black,%
            menucolor=black,%
            pagecolor=black,%
            urlcolor=black,%
            plainpages=false,%
            bookmarksopen=true,% store in PDF document that bookmarks are shown when opening it.
            bookmarksnumbered=true % number bookmarks according to sectioning commands
           ]{hyperref}

\newtheorem{theorem}{Theorem}
\newtheorem{lemma}[theorem]{Lemma}

\newtheorem{claim}{Claim}

\newcommand{\dx}{\ensuremath{\delta_x}}
\newcommand{\dy}{\ensuremath{\delta_y}}
\newcommand{\dz}{\ensuremath{\delta_z}}
\newcommand{\FF}{\ensuremath{\mathbb{F}}}

\gasset{Nfill=y,Nw=.3,Nh=.3,Nmr=.3}
\input{hexnetworksbox.tex}

\newcounter{mycommentcounter}
\begin{document}
\baselineskip 12pt
\suppressfloats[t]

\begin{center}
{\huge\bf Energy Benefit of\\ Network Coding for Multiple Unicast in Wireless Networks } \\
\vspace{5mm}
\begin{tabular}{cc}  Jasper Goseling & Jos H.\ Weber \\
IRCTR/CWPC, WMC Group & IRCTR/CWPC, WMC Group \\
Delft University of Technology & Delft University of Technology \\
The Netherlands & The Netherlands \\
\verb+j.goseling@tudelft.nl+ & \verb+j.h.weber@tudelft.nl+
\end{tabular}
\end{center}

%%%%%%%%%%%%%%%%%%%%%%%%%%%%%%%%%%%%%%%%%%%%%%%%%%%%%%
%
%
%
%%%%%%%%%%%%%%%%%%%%%%%%%%%%%%%%%%%%%%%%%%%%%%%%%%%%%%
\begin{abstract}
\noindent We show that the maximum possible energy benefit of network coding for multiple unicast on wireless networks is at least $3$. This improves the previously known lower bound of $2.4$ from~\cite{effros06tiling}.
\end{abstract}

%%%%%%%%%%%%%%%%%%%%%%%%%%%%%%%%%%%%%%%%%%%%%%%%%%%%%%
%
%
%
%%%%%%%%%%%%%%%%%%%%%%%%%%%%%%%%%%%%%%%%%%%%%%%%%%%%%%
\section{Introduction} \label{sec:intro}
Traditional routing solutions for communication networks keep independent streams of data separate. The idea of network coding is to allow nodes in the network to combine independent data streams. Some of the benefits of network coding that have been demonstrated are increased throughput, reduced resource consumption and increased security, see e.g.~\cite{ncprimer} and the references therein. Our interest is in the reduction in energy consumption in wireless networks offered by network coding. 

The potential of network coding to reduce energy consumption 
%The reduction in energy consumption offered by network coding for wireless networks
 is demonstrated using the example given in Figures~\ref{fig:linerouting} and~\ref{fig:linenc} in which nodes $A$ and $C$ need to exchange bits $x$ and $y$. Figure~\ref{fig:linerouting} shows a routing solution using $4$ transmissions, which is the minimum possible number if only routing is allowed. One can observe that in this case transmissions $\textcircled{\scriptsize 1}$ and $\textcircled{\scriptsize 2}$ are useful only to nodes $C$ and $A$, respectively. The network coding solution from Figure~\ref{fig:linenc} uses $3$ transmissions. Network coding allows transmission $\textcircled{\scriptsize 3}$ to be useful for both $A$ and $C$, increasing efficiency. Without network coding $4$ transmissions are required, whereas the network coding solution uses $3$ transmissions. We say that for this example the energy benefit of network coding is $\frac{4}{3}$. 
\begin{figure}
\begin{center}
\begin{minipage}[t]{.45\linewidth}
\begin{center}
\begin{picture}(6,3.5)(-3,-1)
%	\put(-3.5,2.2){\makebox(2,1){\emph{Without network coding}}}
	\put(-3.5,2.2){\makebox(2,1){$\scriptstyle A$}}
	\put(-1,2.2){\makebox(2,1){$\scriptstyle B$}}
	\put(1.5,2.2){\makebox(2,1){$\scriptstyle C$}}
	
%	\put(-4,2){$\scriptstyle t=1$}
%	\put(-4,1){$\scriptstyle t=2$}
%	\put(-4,0){$\scriptstyle t=3$}
%	\put(-4,-1){$\scriptstyle t=4$}

	\node(a1)(-2.5,2){}
	\node(b1)(0,2){}
	\node(c1)(2.5,2){}
	\node(a2)(-2.5,1){}
	\node(b2)(0,1){}
	\node(c2)(2.5,1){}
	\node(a3)(-2.5,0){}
	\node(b3)(0,0){}
	\node(c3)(2.5,0){}
	\node(a4)(-2.5,-1){}
	\node(b4)(0,-1){}
	\node(c4)(2.5,-1){}
	
	\drawedge(a1,b1){$\scriptstyle x$}
	\drawedge[ELpos=60](b2,c2){$\scriptstyle x$}
	\drawedge[ELpos=25,linewidth=0,AHnb=0](b2,c2){$\textcircled{$\scriptstyle 1$}$}
	\drawedge[ELside=r](c3,b3){$\scriptstyle y$}
	\drawedge[ELside=r,ELpos=60](b4,a4){$\scriptstyle y$}
	\drawedge[ELside=r,ELpos=25,linewidth=0,AHnb=0](b4,a4){$\textcircled{$\scriptstyle 2$}$}
\end{picture}
\caption{Routing solution exchanging bits $x$ and $y$ between nodes $A$ and $C$. Transmissions $\textcircled{\scriptsize 1}$ and $\textcircled{\scriptsize 2}$ are only useful to nodes $C$ and $A$ respectively.\label{fig:linerouting}}
\end{center}
\end{minipage}
\hspace{.05\linewidth}
\begin{minipage}[t]{.45\linewidth}
\begin{center}
\begin{picture}(6,3.5)(-3,-1)
%	\put(1.5,2.2){\makebox(2,1){\emph{With network coding}}}
	\put(-3.5,2.2){\makebox(2,1){$\scriptstyle A$}}
	\put(-1,2.2){\makebox(2,1){$\scriptstyle B$}}
	\put(1.5,2.2){\makebox(2,1){$\scriptstyle C$}}
	
	\node(d1)(-2.5,1.5){}
	\node(e1)(0,1.5){}
	\node(f1)(2.5,1.5){}
	\node(d2)(-2.5,0.5){}
	\node(e2)(0,0.5){}
	\node(f2)(2.5,0.5){}
	\node(d3)(-2.5,-.5){}
	\node[NLdist=.5,NLangle=270](e3)(0,-.5){$\textcircled{$\scriptstyle 3$}$}
	\node(f3)(2.5,-.5){}
	
	\drawedge(d1,e1){$\scriptstyle x$}
	\drawedge[ELside=r](f2,e2){$\scriptstyle y$}
	\drawedge(e3,f3){$\scriptstyle x+y$}
	\drawedge[ELside=r](e3,d3){$\scriptstyle x+y$}
\end{picture}
\caption{Network coding solution. Transmission $\textcircled{\scriptsize 3}$ is useful for both $A$ and $C$. The benefit of network coding for this configuration is $\frac{4}{3}$. \label{fig:linenc}}
\end{center}
\end{minipage}
\end{center}
\end{figure}

The energy benefit of network coding depends on the network topology and the traffic pattern. One can e.g.\ show that the network obtained by extending the network from the previous example to a network of many nodes on a line, allows network coding to reduce energy consumption by a factor $2$. This example was first presented in~\cite{wu04techrep}, where the network coding benefits w.r.t.\ throughput where discussed, the energy benefit, however, follows easily. It was shown in~\cite{effros06tiling} that there exist networks for which this factor is $2.4$.
%with a network coding benefit of $2.4$.
Our aim is to find the maximum possible energy benefit that network coding can offer for multiple unicast traffic in wireless networks. The contribution in this work is a new lower bound of $3$ to this benefit.

In Section~\ref{sec:model} we define the network and traffic model that we use and state our problem more precisely.
%Section~\ref{sec:problem} presents the problem statement.
An overview of known results in the literature is given in Section~\ref{sec:previous} after which we present our result in Section~\ref{sec:result}. Section~\ref{sec:construction} is used to prove this result. Section~\ref{sec:discussion} provides a discussion on the obtained results and possible future work.

%%%%%%%%%%%%%%%%%%%%%%%%%%%%%%%%%%%%%%%%%%%%%%%%%%%%%%
%
%
%
%%%%%%%%%%%%%%%%%%%%%%%%%%%%%%%%%%%%%%%%%%%%%%%%%%%%%%
\section{Model and Problem Statement} \label{sec:model}
Time is slotted. To simplify notation in Section~\ref{sec:construction} we allow nodes to transmit more than once in each time slot. Alternatively we could have rescaled time such that only one transmission from each node occurs in a time slot. All transmissions in the network are broadcasts, i.e.\ transmissions are received by all neighbours. The neighbours of a node are all other nodes in the network that are within a transmission range that is equal and fixed for all nodes in the network.

Transmission is noiseless, no errors occur and there is no interference at the receivers. Although interference does occur in realistic networks, we do not take it into account here. If interference would be part of the model, not all nodes could transmit in the same time slot, at the expense of the throughput, but the number of transmissions that is required would be the same. Since we are not interested in throughput, but only in energy consumption in the network, we do not take interference into account.

The traffic pattern that we consider is multiple unicast. All symbols are from the field $\FF_2$, i.e. they are bits and addition corresponds to the xor operation. The source of each unicast connection has a sequence of source symbols that need to be delivered to the corresponding receiver. For a source $x$, e.g., we have
\begin{equation*} %\label{eq:seqdef}
	x = \begin{bmatrix}\ \dots & x(t-1) & x(t) & x(t+1) & \dots\ \end{bmatrix},
%	$x = \left[ \dots  x(t-1)\  x(t)\  x(t+1)\  \dots\ \right]$,
\end{equation*}
with $x(t)=0$, for $t\leq 0$.
We call a network together with a set of unicast connections a configuration.

%%%%%%%%%%%%%%%%%%%%%%%%%%%%%%%%%%%%%%%%%%%%%%%%%%%%%%
%
%
%
%%%%%%%%%%%%%%%%%%%%%%%%%%%%%%%%%%%%%%%%%%%%%%%%%%%%%%
%\section{Problem Statement} \label{sec:problem}
We are interested in the energy consumption in the network, which we define as the average over time of the number of transmissions used to deliver one symbol from each unicast connection. The energy benefit of network coding for a configuration is defined as the ratio of
the minimum energy consumption of any routing solution and the minimum energy consumption of any network coding solution, i.e.\
\begin{equation*}
\parbox{7cm}{energy benefit of network coding for a wireless multiple unicast configuration} = 
\frac{\parbox{6cm}{minimum energy consumption of any routing solution}}
{\parbox{6cm}{minimum energy consumption of any network coding solution}}.
\end{equation*}

 In this paper we will refer to this ratio as the energy benefit of network coding, or simply as the benefit of network coding.

%%%%%%%%%%%%%%%%%%%%%%%%%%%%%%%%%%%%%%%%%%%%%%%%%%%%%%
%
%
%
%%%%%%%%%%%%%%%%%%%%%%%%%%%%%%%%%%%%%%%%%%%%%%%%%%%%%%
\newpage
\section{Previous Work}  \label{sec:previous}
The best known lower bound on the maximum energy benefit of network coding over all possible configurations is $2.4$~\cite{effros06tiling}.
%The best known lower bound to the energy benefit of network coding is $2.4$~\cite{effros06tiling}.
The network code that was constructed to show this lower bound, satisfies the property that data symbols transmitted by a node are linear combinations only of source symbols that have been successfully decoded by that node. The rationale behind this is that in this way information in the network can be constrained to the neighbourhood of the path between the source and destination of the corresponding unicast session, a network code design heuristic that was introduced in~\cite{katti06xor}.

In~\cite{liu07scale} it was shown that for random networks the energy benefit of network codes satisfying the above property is upper bounded by $3$. It is not known if $3$ is also an upper bound on the energy benefit of arbitrary network codes on arbitrary configurations, that is allowing codes that do not satisfy the above property on networks with arbitrary topology and traffic requirements. 

In this work we present a new lower bound on the energy benefit of network coding. In the network code that we construct, nodes transmit linear combinations of data symbols, for which the corresponding source symbols have not necessarily been decoded by these nodes. Our code therefore does not satisfy the above property.

%%%%%%%%%%%%%%%%%%%%%%%%%%%%%%%%%%%%%%%%%%%%%%%%%%%%%%
%
%
%
%%%%%%%%%%%%%%%%%%%%%%%%%%%%%%%%%%%%%%%%%%%%%%%%%%%%%%
%\section{Benefit of $3$ on hexagonal lattice} \label{sec:hex3}
\section{Result} \label{sec:result}
We present a result that is based on the configuration that consists of the network in which the nodes are located on the hexagonal lattice with connectivity as depicted in Figure~\ref{fig:hexnetwork}, together with the unicast sessions that are depicted in Figures~\ref{fig:hextrafficx}, \ref{fig:hextrafficy} and~\ref{fig:hextrafficz}. The figures depict a network with $4$ nodes on each edge and $4$ unicast sessions of each type, i.e. $x_i$, $y_i$ and $z_i$, $i=1,\dots,4$. In general, we will consider networks with $K$ nodes on the network edges and $K$ sessions of each type. The network topology that we consider is equal to the one used in~\cite{effros06tiling}. Our traffic pattern, however, is slightly different. We discuss this in more detail in Section~\ref{sec:discussion}.

\begin{figure}
\begin{minipage}[t]{.45\textwidth}
\begin{center}
\input{hexnetwork.tex}
\caption{Network with nodes positioned at hexagonal lattice. Each edge of the network consists of $K$ nodes.\label{fig:hexnetwork}}
\end{center}
\end{minipage}
%\end{figure}
%\begin{figure}
\hspace{.05\textwidth}
\begin{minipage}[t]{.45\textwidth}
\begin{center}
\input{hextrafficx.tex}
\caption{Sources $S(\cdot)$ and receivers $R(\cdot)$ for first set of unicast connections defined on the network from Figure~\ref{fig:hexnetwork}.\label{fig:hextrafficx}}
\end{center}
\end{minipage}
\end{figure}

\begin{figure}
\begin{minipage}[b]{.45\textwidth}
\begin{center}
\input{hextrafficy.tex}
\caption{Second set of unicast connections defined on the network from Figure~\ref{fig:hexnetwork}.\label{fig:hextrafficy}}
\end{center}
\end{minipage}
%\end{figure}
%\begin{figure}
\hspace{.05\textwidth}
\begin{minipage}[b]{.45\textwidth}
\begin{center}
\input{hextrafficz.tex}
\caption{Third set of unicast connections defined on the network from Figure~\ref{fig:hexnetwork}.\label{fig:hextrafficz}}
\end{center}
\end{minipage}
\end{figure}

\begin{lemma} \label{lem:routing}
%\begin{equation}
%	\EEr = 3\NN(K-1) = \frac{3(K-1)K}{2}.
%	\EEr = 3\NN(K-1).
%\end{equation}
The minimum energy consumption of any routing solution for the configuration from Figures~\ref{fig:hexnetwork}--\ref{fig:hextrafficz} is $3\binom{K}{2}$.
\end{lemma}
\begin{proof}
In the minimum cost routing solution, symbols from each unicast connection follows the shortest route from source to receiver, as depicted in Figures~\ref{fig:hextrafficx} --~\ref{fig:hextrafficz}.
\end{proof}

\begin{lemma} \label{lem:nc}
%\begin{equation}
%	\EEnc \leq 3\left( \NN(K)-\NN(K-3) \right) + \NN(K-3) = \frac{1}{2}K^2 + 6\frac{1}{2}K-6.
%	\EEnc \leq 3\NN(K)-2\NN(K-3).
%\end{equation}
For the configuration from Figures~\ref{fig:hexnetwork}--\ref{fig:hextrafficz} there exists a network coding solution that has energy consumption $3\binom{K+1}{2} - 2\binom{K-2}{2}$.
\end{lemma}
We will prove Lemma~\ref{lem:nc} in Section~\ref{sec:construction} by constructing a network code that achieves this bound. The next theorem states our main result.
\begin{theorem}
There exist multiple unicast wireless networks for which network coding offers an energy benefit of $3$.
%\begin{equation}
%	\sup_{K}\frac{\EEr}{\EEnc} \geq 3.
%\end{equation}
\end{theorem}
\begin{proof}
The result follows from Lemmas~\ref{lem:routing} and~\ref{lem:nc} by taking the limit of $K$ to infinity, i.e.
\begin{equation*}
%\lim_{K\to\infty}\frac{\EEr}{\EEnc} \geq
\lim_{K\to\infty} \frac{
3\binom{K}{2}
}{
3\binom{K+1}{2} - 2\binom{K-2}{2}
}
%= \lim_{K\to\infty} \frac{3(K-1)K}{3K(K+1)-2(K-3)(K-2)} = \lim_{K\to\infty}\frac{3K^2}{K^2}
= 3.
\end{equation*}
\end{proof}

%%%%%%%%%%%%%%%%%%%%%%%%%%%%%%%%%%%%%%%%%%%%%%%%%%%%%%
%
%
%
%%%%%%%%%%%%%%%%%%%%%%%%%%%%%%%%%%%%%%%%%%%%%%%%%%%%%%
\section{Network Code Construction} \label{sec:construction}
In this section we construct a network code for which the energy consumption is according to Lemma~\ref{lem:nc}. We first introduce some notation. Let
\begin{equation*}
	\tilde x_i(t) = \sum_{\tau=0}^{i-1} x_{i-\tau}(t-\tau),
\end{equation*}
$i\in\{1,\dots,K\}$, with $\tilde y_j(t)$ and $\tilde z_k(t)$ defined similarly. Also, let $A[P], B[P], \dots, F[P]$ be the neighbours of a node $P$ as depicted in Figure~\ref{fig:hexneighbours}.
\begin{figure}
\begin{minipage}[t]{.45\textwidth}
\begin{center}
\input{hexneighbours.tex}
\caption{The neighbours of a node $P$.\label{fig:hexneighbours}}
\end{center}
\end{minipage}
%\end{figure}
%\begin{figure}
\hspace{.05\textwidth}
\begin{minipage}[t]{.45\textwidth}
\begin{center}
\input{hexinit.tex}
\caption{Values for $i$, $j$, $k$, $\dx$, $\dy$ and $\dz$ for some nodes in the network. Remaining values follow from~\eqref{eq:indexing}.\label{fig:hexinit}}
\end{center}
\end{minipage}
\end{figure}
The code is defined by the following properties:
\begin{enumerate}
%\item In each time slot one information symbol from each connection is decoded. 
\item The data symbols transmitted by nodes are linear combinations of information symbols of the form
\begin{equation*}
	\tilde x_i(t - \dx) + \tilde y_j(t - \dy) + \tilde z_k(t - \dz),
\end{equation*}
where $t\in\mathbb{N}^+$ is the time slot and $\dx, \dy, \dz \in \mathbb{N}$ and $i, j, k \in \{1,\dots, K\}$ are per node constants, i.e. they may be different for each node, but are the same for all symbols transmitted by a specific node. 
%\item If a node transmits $\tilde x_i(t_x) + \tilde y_j(t_y) + \tilde z_k(t_z)$ in one time slot, it transmits $\tilde x_i(t_x+1) + \tilde y_j(t_y+1) + \tilde z_k(t_z+1)$ in the next one.
\item Let $P(t)=\tilde x_i(t-\dx) + \tilde y_j(t-\dy) + \tilde z_k(t-\dz)$ be the symbol sent by node $P$ in time slot $t$. The symbols transmitted by its neighbours in that same time slot are
%\begin{align}
%	A[P](t) &= \tilde x_{i+1}(t) + \tilde y_{j-1}(t+1) + \tilde z_{k}(t-1) \\
%	B[P](t) &= \tilde x_{i+1}(t-1) + \tilde y_{j}(t+1) + \tilde z_{k-1}(t) \\
%	C[P](t) &= \tilde x_{i}(t-1) + \tilde y_{j+1}(t) + \tilde z_{k-1}(t-1) \\
%	D[P](t) &= \tilde x_{i-1}(t) + \tilde y_{j+1}(t-1) + \tilde z_{k}(t+1) \\
%	E[P](t) &= \tilde x_{i-1}(t+1) + \tilde y_{j}(t-1) + \tilde z_{k+1}(t) \\
%	F[P](t) &= \tilde x_{i}(t+1) + \tilde y_{j-1}(t) + \tilde z_{k+1}(t-1)
%\end{align}
\begin{equation} \label{eq:indexing}
\begin{array}{lllllll}
	A[P](t) & = & \tilde x_{i+1}(t-\dx) & + & \tilde y_{j-1}(t-\dy+1) & + & \tilde z_{k}(t-\dz-1), \\
	B[P](t) & = & \tilde x_{i+1}(t-\dx-1) & + & \tilde y_{j}(t-\dy+1) & + & \tilde z_{k-1}(t-\dz), \\
	C[P](t) & = & \tilde x_{i}(t-\dx-1) & + & \tilde y_{j+1}(t-\dy) & + & \tilde z_{k-1}(t-\dz+1), \\
	D[P](t) & = & \tilde x_{i-1}(t-\dx) & + & \tilde y_{j+1}(t-\dy-1) & + & \tilde z_{k}(t-\dz+1), \\
	E[P](t) & = & \tilde x_{i-1}(t-\dx+1) & + & \tilde y_{j}(t-\dy-1) & + & \tilde z_{k+1}(t-\dz), \\
	F[P](t) & = & \tilde x_{i}(t-\dx+1) & + & \tilde y_{j-1}(t-\dy) & + & \tilde z_{k+1}(t-\dz-1).
\end{array}
\end{equation}
%\end{enumerate}
%Note that instead of $P(t)=\tilde x_i(t) + \tilde y_j(t) + \tilde z_k(t)$ we would in general have $P(t)=\tilde x_i(t+t_x) + \tilde y_j(t+t_y) + \tilde z_k(t+t_y)$ with $t_x, t_y, t_z\in\mathbb{Z}$. Since, however, we are only interested in the relative indexing expressed by~\eqref{eq:indexing} and all sequences are unbounded, the variables $t_x$, $t_y$ and $t_z$ can be assumed $0$ w.l.o.g. The absolute time indices follow from
%\begin{enumerate}
%\setcounter{enumi}{3}
\item The exception to the above two rules comes from all nodes that are at an edge or corner of the network. These nodes transmit three data symbols in each time slot. If~\eqref{eq:indexing} dictates that a node should transmit $P(t)=\tilde x_i(t-\dx) + \tilde y_j(t-\dy) + \tilde z_k(t-\dz)$, and the node is at an edge or corner, it transmits three different symbols: $P_x(t)=\tilde x_i(t-\dx)$, $P_y(t)=\tilde y_j(t-\dy)$ and $P_z(t)=\tilde z_k(t-\dz)$. For notational convenience later on let
\begin{equation} \label{eq:psum}
	P(t)=P_x(t)+P_y(t)+P_z(t).
\end{equation}
Note that $P(t)$ is not actually transmitted by nodes at edges or corners of the network, but only a notational shortcut.
\item Let $R$ be the receiver of source $z_k$, i.e. $R$ is a node on the left edge of the network. 
%It can be verified that in time slot $t$, $R$ transmits $R_z(t)=\tilde z_k(t+1-K)$.
Suppose that in time slot $t$, $R$ transmits $R_z(t)=\tilde z_k(t-\dz)$. In that same time slot node $R$ decodes source symbol $z_k(t-\dz)$. This implies that, after time slot $\dz$, one source symbol from $z_k$ is decoded each time slot. Similar decoding procedures are used at all other receivers.
\item The only thing that remains to be specified is the value of $i$, $j$, $k$, $\dx$, $\dy$ and $\dz$ for all nodes. We only specify some values at the corners of the network. These are given in Figure~\ref{fig:hexinit}. The remaining values follow from~\eqref{eq:indexing}.
\end{enumerate}

We need to show that the scheme is valid, i.e. that all nodes are able to produce the required linear combinations and that all receivers are able to decode. We need to analyze different cases, depending on the location of the node. We distinguish between nodes that are at corners of the network, at edges of the network and the remaining nodes, which we refer to as internal nodes. Since our construction is symmetric and homogeneous we will consider only the node in the top corner, an arbitrary node at the left edge, and an arbitrary internal node.
%By symmetry of the scheme we need to consider only three nodes, node $P$, an internal node, node $Q$ at the left edge and node $R$ in the top corner of the network.
\begin{claim} \label{claim:internal}
Let $P$ be any internal node. The symbol $P(t+1)$ transmitted by $P$ in time slot $t+1$ satisfies
\begin{multline}
	P(t+1) = A[P](t-1) + B[P](t) + \\ C[P](t-1) + D[P](t) + E[P](t-1) + F[P](t) + P(t-2).
\end{multline}
\end{claim}
\begin{proof}
Assume that $P(t)=\tilde x_i(t-\dx) + \tilde y_j(t-\dy) + \tilde z_k(t-\dz)$. We have
\begin{equation*}
	A[P](t-1)+B[P](t) = y_j(t-\dy+1) + \tilde z_k(t-\dz-2) + \tilde z_{k-1}(t-\dz).
\end{equation*}
\begin{equation*}
	C[P](t-1)+D[P](t) = \tilde x_i(t-\dx-2) + \tilde x_{i-1}(t-\dx) + z_{k}(t-\dz+1).
\end{equation*}
\begin{equation*}
	E[P](t-1)+F[P](t) = x_i(t-\dx+1) + \tilde y_j(t-\dy-2) + \tilde y_{j-1}(t-\dy).
\end{equation*}
The result follows from $\tilde x_i(t-\dx+1) = x_i(t-\dx+1)+\tilde x_{i-1}(t-\dx)$ and equivalent relations for $\tilde y_j(t-\dy+1)$ and $\tilde z_k(t-\dz+1)$. Note that if some of P's neighbours are on the border of the network we require~\eqref{eq:psum}.
\end{proof}

\begin{claim} \label{claim:border}
Let $Q$ be any node on the left edge of the network. Assume $Q_x(t)=\tilde x_i(t-\dx)$. The symbols transmitted by $Q$ in time slot $t+1$ satisfy 
\begin{align}
	Q_x(t+1) =\ &x_i(t+1-\dx) + E[Q]_x(t-1), \label{eq:qx} \\
	Q_y(t+1) =\ &B[Q]_y(t)
\intertext{and}
	Q_z(t+1) =\ &B[Q]_z(t)+C[Q](t-1)+D[Q](t)+ \notag \\ 
	                 &E[Q]_x(t-1)+Q_x(t-2).
\end{align}
\end{claim}
\begin{proof}
Assume that $Q_y(t)=\tilde y_j(t-\dy)$ and $Q_z(t)=\tilde z_k(t-\dz)$. Since $Q$ is on the left edge of the network it has only neighbours $B[Q]$, $C[Q]$, $D[Q]$ and $E[Q]$. Noting that $x_i(t+1-\dx)$ is available as a source symbol, the relations for $Q_x(t+1)$ and $Q_y(t+1)$ are readily verified. For $Q_z(t+1)$ we consider
\begin{equation*}
	C[Q](t-1)+D[Q](t) = \tilde x_i(t-\dx-2) + \tilde x_{i-1}(t-\dx) + z_k(t-\dz+1)
\end{equation*}
and note that $B[Q]_z(t)=\tilde z_{k-1}(t-\dz)$ and $E[Q]_x(t-1)=\tilde x_{i-1}(t-\dx)$.
\end{proof}

\begin{claim} \label{claim:corner}
Let $R$ be the node in the top corner of the network. The symbols transmitted by $R$ in time slot $t+1$ satisfy
\begin{align}
	R_x(t+1) =\ &x_K(t+2-K) + E[Q]_x(t-1), \label{eq:rx} \\
	R_y(t+1) =\ &y_1(t+1) \label{eq:ry}
\intertext{and}
	R_z(t+1) =\ &D_z[R](t). \label{eq:rz}
\end{align}
%Given node $R$ in the top corner, i.e. $R$ has only neighours $D[R]$ and $E[R]$, we have $R_x(t+1)$ and $R_y(t+1)$ available as source symbols and $R_z(t+1) = D[R]_z(t)$.	
\end{claim}
\begin{proof}
Note that from~\eqref{eq:indexing} and Figure~\ref{fig:hexinit} it follows that for $R$: $i=K$, $j=1$, $\dx=K-1$ and $\dy=0$. The proof of~\eqref{eq:rx} is equivalent to the proof of~\eqref{eq:qx}. Relations~\eqref{eq:ry} and~\eqref{eq:rz} can  be easily verified.
\end{proof}

\begin{proof}[Proof of Lemma~\ref{lem:nc}]
First we need to prove that the scheme is valid. From Claims~\ref{claim:internal} --~\ref{claim:corner}  it follows that all nodes can produce the required linear combinations of symbols to transmit. We also need to show that source symbols can be decoded. Consider node $Q$ on the left edge of the network that is the receiver of $z_k$. Suppose that in time slot $t$ it transmits $Q_z(t) = \tilde z_k(t-\dz)$. It can recover $z_k(t-\dz)$ as
\begin{equation*}
	z_k(t-\dz) = Q_z(t) + B[Q]_z(t-1).
\end{equation*}
Node $R$ in the top corner does not have a neighbour $B[R]$, but it needs to decode $z_1(t-\dz)$ for which $ \tilde z_1(t-\dz) = R_z(t) = z_1(t-\dz)$.

The contributions to the energy consumption are $1$ for each of the $\binom{K-2}{2}$ internal nodes and $3$ for each of the $\binom{K+1}{2} - \binom{K-2}{2}$ nodes at the border of the network.
\end{proof}

%%%%%%%%%%%%%%%%%%%%%%%%%%%%%%%%%%%%%%%%%%%%%%%%%%%%%%
%
%
%
%%%%%%%%%%%%%%%%%%%%%%%%%%%%%%%%%%%%%%%%%%%%%%%%%%%%%%
\section{Discussion} \label{sec:discussion}
We have shown that the energy benefit of network coding for multiple unicast in wireless networks is at least $3$. The network topology that we use in our constructive proof is the same as used in~\cite{effros06tiling} in which a lower bound of $2.4$ was obtained. The difference is in the traffic pattern used.
%The three sets of unicast patterns from Figures~\ref{fig:hextrafficx}--\ref{fig:hextrafficz} are PACKED DENSE.
In~\cite{effros06tiling} the number of unicast sessions  is smaller than considered in this paper. It is mentioned in~\cite{effros06tiling} that for the number of unicast sessions used in this paper there does not seem to be a valid network code for which: 1)  internal nodes in the network transmit only once per decoded source symbol and 2) data symbols transmitted by a node are linear combinations only of source symbols that have been successfully decoded by that node, i.e. satisfying the property discussed in Section~\ref{sec:previous}.

We have obtained a valid code satisfying 1), but not 2). It is not known if a code satisfying both properties exists.

The only known upper bound to the energy benefit of network coding for multiple unicast in wireless networks comes from~\cite{liu07scale}, in which only a restricted class of network codes is considered. It is an open problem to find general upper bounds.

%%%%%%%%%%%%%%%%%%%%%%%%%%%%%%%%%%%%%%%%%%%%%%%%%%%%%%
%
%
%
%%%%%%%%%%%%%%%%%%%%%%%%%%%%%%%%%%%%%%%%%%%%%%%%%%%%%%
\bibliographystyle{IEEEtran}
\bibliography{IEEEabrv,networkcoding}

\end{document}

%% file: hexnetworksbox.tex
\newsavebox{\hexnetwork}
\savebox{\hexnetwork}{
\gasset{Nfill=y,Nw=.3,Nh=.3,Nmr=.3}

\node(x4y1)(0,3.46){}

\node(x3y1)(-1,1.73){}
\node(x3y2)(1,1.73){}

\node(x2y1)(-2,0){}
\node(x2y2)(0,0){}
\node(x2y3)(2,0){}

\node(x1y1)(-3,-1.73){}
\node(x1y2)(-1,-1.73){}
\node(x1y3)(1,-1.73){}
\node(x1y4)(3,-1.73){}
}

%% file: hexnetwork.tex
\begin{picture}(6,5.5)(-3,-2)
\put(0,0){\usebox{\hexnetwork}}

{
\gasset{AHnb=1,ATnb=1,exo=-.5,sxo=-.5}
\drawedge(x1y1,x4y1){$\scriptstyle K$}
}

{
\gasset{AHnb=0,ATnb=0}

\drawedge(x1y1,x1y2){}
\drawedge(x1y2,x1y3){}
\drawedge(x1y3,x1y4){}

\drawedge(x2y1,x2y2){}
\drawedge(x2y2,x2y3){}

\drawedge(x3y1,x3y2){}

\drawedge(x4y1,x3y1){}
\drawedge(x3y1,x2y1){}
\drawedge(x2y1,x1y1){}

\drawedge(x3y2,x2y2){}
\drawedge(x2y2,x1y2){}

\drawedge(x2y3,x1y3){}

\drawedge(x1y4,x2y3){}
\drawedge(x2y3,x3y2){}
\drawedge(x3y2,x4y1){}

\drawedge(x1y3,x2y2){}
\drawedge(x2y2,x3y1){}

\drawedge(x1y2,x2y1){}
}

\end{picture}

%% file: hextrafficx.tex
\begin{picture}(6,5.5)(-3,-2)

\put(0,0){\usebox{\hexnetwork}}

{\gasset{NLangle=180,NLdist=.9}
\nodelabel(x1y1){$\scriptstyle S(x_1)$}
\nodelabel(x2y1){$\scriptstyle S(x_2)$}
\nodelabel(x3y1){$\scriptstyle S(x_3)$}
\nodelabel(x4y1){$\scriptstyle S(x_4)$}
}

{\gasset{NLangle=0,NLdist=.9}
\nodelabel(x1y4){$\scriptstyle R(x_1)$}
\nodelabel(x2y3){$\scriptstyle R(x_2)$}
\nodelabel(x3y2){$\scriptstyle R(x_3)$}
\nodelabel(x4y1){$\scriptstyle R(x_4)$}
}

{
\gasset{AHnb=0,ATnb=0}

\drawedge(x1y1,x1y2){}
\drawedge(x1y2,x1y3){}
\drawedge[AHnb=1](x1y3,x1y4){}

\drawedge(x2y1,x2y2){}
\drawedge[AHnb=1](x2y2,x2y3){}

\drawedge[AHnb=1](x3y1,x3y2){}
}

\end{picture}

%% file: hextrafficy.tex
\begin{picture}(6,6.8)(-3,-2.5)

\put(0,0){\usebox{\hexnetwork}}

{
\gasset{Nfill=y,Nw=.3,Nh=.3,Nmr=.3,NLangle=240,NLdist=.7}

\nodelabel(x1y1){$\scriptstyle R(y_1)$}
\nodelabel(x1y2){$\scriptstyle R(y_2)$}
\nodelabel(x1y3){$\scriptstyle R(y_3)$}
\nodelabel(x1y4){$\scriptstyle R(y_4)$}

}

{
\gasset{Nfill=n,Nw=.3,Nh=.3,Nmr=.3,NLangle=	60,NLdist=.7}

\nodelabel(x4y1){$\scriptstyle S(y_1)$}
\nodelabel(x3y2){$\scriptstyle S(y_2)$}
\nodelabel(x2y3){$\scriptstyle S(y_3)$}
\nodelabel(x1y4){$\scriptstyle S(y_4)$}

}

{
\gasset{AHnb=0,ATnb=0}

\drawedge(x4y1,x3y1){}
\drawedge(x3y1,x2y1){}
\drawedge[AHnb=1](x2y1,x1y1){}

\drawedge(x3y2,x2y2){}
\drawedge[AHnb=1](x2y2,x1y2){}

\drawedge[AHnb=1](x2y3,x1y3){}
}

\end{picture}

%% file: hextrafficz.tex
\begin{picture}(6,6.8)(-3,-2.5)

\put(0,0){\usebox{\hexnetwork}}

{
\gasset{Nfill=y,Nw=.3,Nh=.3,Nmr=.3,NLangle=120,NLdist=.7}

\nodelabel(x4y1){$\scriptstyle R(z_1)$}
\nodelabel(x3y1){$\scriptstyle R(z_2)$}
\nodelabel(x2y1){$\scriptstyle R(z_3)$}
\nodelabel(x1y1){$\scriptstyle R(z_4)$}
}

{
\gasset{Nfill=n,Nw=.3,Nh=.3,Nmr=.3,NLangle=	300,NLdist=.7}

\nodelabel(x1y4){$\scriptstyle S(z_1)$}
\nodelabel(x1y3){$\scriptstyle S(z_2)$}
\nodelabel(x1y2){$\scriptstyle S(z_3)$}
\nodelabel(x1y1){$\scriptstyle S(z_4)$}

}

{
\gasset{AHnb=0,ATnb=0}

\drawedge(x1y4,x2y3){}
\drawedge(x2y3,x3y2){}
\drawedge[AHnb=1](x3y2,x4y1){}

\drawedge(x1y3,x2y2){}
\drawedge[AHnb=1](x2y2,x3y1){}

\drawedge[AHnb=1](x1y2,x2y1){}
}

\end{picture}

%% file: hexneighbours.tex
\begin{picture}(6,6.5)(-3,-3)

{
\gasset{Nfill=y,Nw=.3,Nh=.3,Nmr=.3,NLangle=0,NLdist=.7}

%\node(x4y1)(0,2){}

\node(x3y1)(-1,1.73){$\scriptstyle A[P]$}
\node(x3y2)(1,1.73){$\scriptstyle B[P]$}

\node(x2y1)(-2,0){$\scriptstyle F[P]$}
\node(x2y2)(0,0){$\scriptstyle P$}
\node(x2y3)(2,0){$\scriptstyle C[P]$}

%\node(x1y1)(-3,-1){}
\node(x1y2)(-1,-1.73){$\scriptstyle E[P]$}
\node(x1y3)(1,-1.73){$\scriptstyle D[P]$}
%\node(x1y4)(3,-1){$\scriptstyle R(x_3)$}

}

\end{picture}

%% file: hexinit.tex
\begin{picture}(6,6.5)(-3,-2.5)

\put(0,0){\usebox{\hexnetwork}}

{\gasset{NLangle=270,NLdist=.6}
\nodelabel(x1y1){$\scriptstyle i=1,~\dx=0$}
%\nodelabel(x1y1){$\scriptstyle P_x(1) = \tilde x_1(1)$}
%\nodelabel(x2y1){$\scriptstyle \tilde x_2(0)$}
%\nodelabel(x3y1){$\scriptstyle \tilde x_3(-1)$}
}

{\gasset{NLangle=0,NLdist=1.5}
\nodelabel(x4y1){$\scriptstyle j=1,~\dy=0$}
%\nodelabel(x4y1){$\scriptstyle P_y(1) = \tilde y_1(1)$}
%\nodelabel(x3y2){$\scriptstyle \tilde y_2(0)$}
%\nodelabel(x2y3){$\scriptstyle \tilde y_3(-1)$}
}

{\gasset{NLangle=270,NLdist=.6}
\nodelabel(x1y4){$\scriptstyle k=1,~\dz=0$}
%\nodelabel(x1y4){$\scriptstyle P_z(1) = \tilde z_1(1)$}
%\nodelabel(x1y3){$\scriptstyle \tilde z_2(0)$}
%\nodelabel(x1y2){$\scriptstyle \tilde z_3(-1)$}
}

\end{picture}